\documentclass[11pt]{llncs}

\usepackage[latin1]{inputenc}
\usepackage{amsmath}
\usepackage{amssymb}
\usepackage{enumerate}
\usepackage{xspace}
\usepackage{amsfonts}
\usepackage{mathrsfs}
\usepackage[all,2cell]{xy}
\usepackage{hyperref}
\CompileMatrices 
\UseAllTwocells
\SilentMatrices
\usepackage{graphicx}
\usepackage{cite}
\usepackage{float}
\usepackage{fancybox}
\usepackage{mathtools}
\usepackage{stmaryrd}
\usepackage{wrapfig}
\DeclareMathOperator{\op}{op}

\DeclareMathOperator{\id}{id}

\DeclareMathOperator{\tw}{\mathit{tw}}

\renewcommand{\P}{\mathcal{P}}

\newcommand{\Defeq}
 {\stackrel{\mathrm{def}}{=}}

\newcommand{\stran}{\raise1pt\hbox{$\centerdot$}}

\newcommand{\rring}[1]{\ensuremath{\mathbb{#1}}}

\newcommand{\N}{\rring{N}}

\newcommand{\Set}{\cat{Set}\xspace}

\newcommand{\diag}{\Delta}

\newcommand{\codiag}{\triangledown}

\renewcommand{\t}{\times}

\newcommand{\ladj}[2]{\ar@/^/[#1]^-{#2} \ar@{}[#1]|-%

{\ifthenelse{\equal{#1}{r}}{\bot}{%

{\ifthenelse{\equal{#1}{rr}}{\bot}{%

{\ifthenelse{\equal{#1}{l}}{\top}{%

{\ifthenelse{\equal{#1}{u}}{\dashv}{%

{\vdash}}}}}}}}}}

\newcommand{\radj}[2]{\ar@/^/[#1]^-{#2}}

\newcommand{\radjff}[2]{\ar@{_{(}->}[#1]^{#2}}

\newcommand{\pullbacktop}[4]{%

{#1} \ar@/_/[ddr]_{#4} \ar@/^/[drr]^{#2}%

\ar@{.>}[dr]|-{#3} \\}

\newtheorem{clm}{Claim}[section]

\newtheorem{lem}[clm]{Lemma}

\newtheorem{rmk}[clm]{Remark}

\newcommand{\cat}[1]{\ensuremath{\mathbf{#1}}}

\newcounter{ncomm}

\newcommand{\ltsred}[1]
{ \setbox0=\hbox{$\ {}^{#1}\ $}
  \setbox1=\hbox{$\longrightarrow$}
  \loop\setbox1=\hbox{$-$\kern-0.3em\unhbox1}\ifdim\wd1<\wd0\repeat
  \hbox{$\ \ \mathop{\box1}\limits^{#1}\ \ $}
}

\newcommand{\arx}[2]{\!\xymatrix@=15pt{\ar[r]^{{#1}}_{{#2}}&}\!}

\newlength{\mylength}

{\setlength{\fboxsep}{15pt} 
\setlength{\mylength}{\linewidth}%
\addtolength{\mylength}{-2\fboxsep}%
\addtolength{\mylength}{-2\fboxrule}%
\Sbox 
\minipage{\mylength}%
\setlength{\abovedisplayskip}{0pt}%
\setlength{\belowdisplayskip}{0pt}%
}%
{\endminipage\endSbox 
\[\fbox{\TheSbox}\]}

\newcommand{\ind}[1]{\mathrel{\shortparallel_{#1}}}
\newcommand{\Rel}{\mathbf{Rel}}
\newcommand{\lift}[1]{{#1}^{\#}}
\newcommand{\contsymb}{\bowtie}
\newcommand{\contention}[1]{\mathrel{\contsymb_{#1}}}
\newcommand{\comp}{\mathrel{;}}
\newcommand{\multiset}[1]{\mathcal{M}_{#1}}
\newcommand{\U}{\mathcal{U}}
\newcommand{\V}{\mathcal{V}}

\newcommand{\msynch}[2]{\mathsf{snc}(#1,#2)}
\newcommand{\minmsynch}[2]{\mathsf{minsnc}(#1,#2)}

\newcommand{\Csp}[1]{\mathsf{Csp}(#1)}
\newcommand{\Sp}[1]{\mathsf{Sp}(#1)}
\newcommand{\Spr}[1]{\mathsf{Spr}(#1)}
\newcommand{\figref}[1]{Fig.\;\ref{#1}}
\renewcommand{\diag}{\lower12pt\hbox{$\includegraphics[width=1cm]{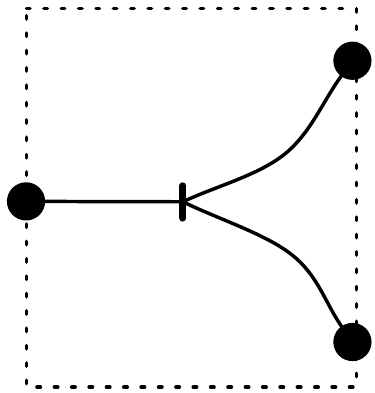}$}\,}
\renewcommand{\codiag}{\lower12pt\hbox{$\includegraphics[width=1cm]{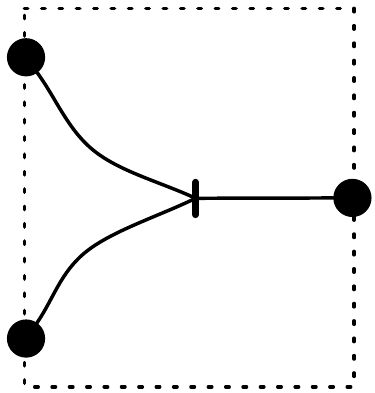}$}\,}
\newcommand{\ldiag}{\lower12pt\hbox{$\includegraphics[width=1cm]{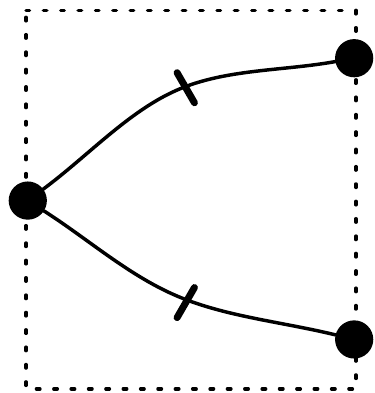}$}\,}
\newcommand{\lcodiag}{\lower12pt\hbox{$\includegraphics[width=1cm]{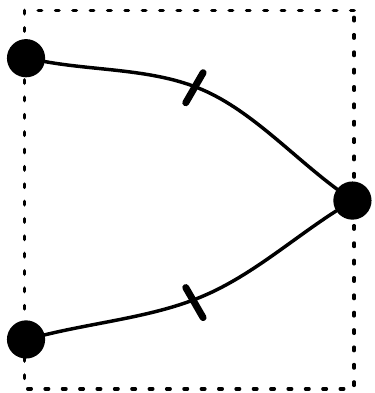}$}\,}
\newcommand{\Setfc}{\Set_f^{c}}
\newcommand{\Relfc}{\Rel_f^{c}}
\newcommand{\multirel}{\Rel^{\mathcal{M}}}
\newcommand{\Relfm}{\Rel_f^{\mathcal{M}}}
\newcommand{\Pc}{\P_{c}}
\newcommand{\Kl}[1]{\mathsf{Kl}(#1)}
\makeatletter
\newcommand{\superimpose}[2]{%
  {\ooalign{$#1\@firstoftwo#2$\cr\hfil$#1\@secondoftwo#2$\hfil\cr}}}
\makeatother
\newcommand{\relto}{\mathrel{\mathpalette\superimpose{{\rightarrow}{\shortmid}}}}
\newcommand{\mrelto}{\mathrel{\mathpalette\superimpose{{\rightarrow}{\!\raise.3pt\hbox{$\shortparallel$}}}}}
\newcommand{\minsynch}[2]{\mathsf{minsnc}(#1,#2)}
\renewcommand{\op}{\mathsf{op}}
\newcommand{\xrelleftarrow}[1]{\mathrel{\mathpalette\superimpose{{\xleftarrow{#1}}{\shortmid}}}}
\newcommand{\xrelrightarrow}[1]{\mathrel{\mathpalette\superimpose{{\xrightarrow{#1}}{\shortmid}}}}
\newcommand{\xmrelleftarrow}[1]{\mathrel{\mathpalette\superimpose{{\xleftarrow{#1}}{\,\raise.3pt\hbox{$\shortparallel$}}}}}
\newcommand{\xmrelrightarrow}[1]{\mathrel{\mathpalette\superimpose{{\xrightarrow{#1}}{\!\raise.3pt\hbox{$\shortparallel$}}}}}
\newcommand{\synch}[2]{\langle #1 \curlyveedownarrow #2 \rangle}

\newcommand{\idsym}{\mathsf{I}}
\newcommand{\diagsym}{\mathsf{\Delta}}
\newcommand{\codiagsym}{
\ensuremath{\mathchoice
{{\rotatebox[origin=c]{180}{$\diagsym$}}\displaystyle} 
{{\rotatebox[origin=c]{180}{$\diagsym$}}\textstyle} 
{\!{\rotatebox[origin=c]{180}{$\scriptstyle\diagsym$}}\!\scriptstyle} 
{{\rotatebox[origin=c]{180}{$\scriptscriptstyle\diagsym$}}\scriptscriptstyle}
}}
\newcommand{\ldiagsym}{{\mathsf{\Lambda}}}
\newcommand{\lcodiagsym}{{\mathsf{V}}}
\newcommand{\leftEndsym}{\ensuremath{\pmb{\top}}}
\newcommand{\rightEndsym}{\ensuremath{\pmb{\bot}}}
\newcommand{\twsym}{\ensuremath{\mathsf{X}}}
\newcommand{\lzerosym}{\ensuremath{\,\pmb{\uparrow}\,}}
\newcommand{\rzerosym}{\ensuremath{\,\pmb{\downarrow}\,}}
\newcommand{\Setf}{\Set_f}

\title{Nets, relations and linking diagrams}
\author{Pawe{\l} Soboci{\'n}ski}
\institute{ECS, University of Southampton, UK}
\begin{document}
\maketitle
\begin{abstract}
 In recent work, the author and others have studied compositional algebras of Petri nets. Here we consider mathematical aspects of the pure linking algebras that underly them. We characterise composition of nets without places as the composition of spans over appropriate categories of relations, and study the underlying algebraic structures.
\end{abstract}


\section*{Introduction}


Linking structures are ubiquitous in Computer Science, Logic and Mathematics. Amongst many examples, we mention Kelly-Laplaza graphs for compact closed categories~\cite{Kelly1980} and proof nets~\cite{Girard1987}. Linking diagrams\footnote{We use this terminology loosely to mean ``string diagrams without boxes.''} underly string diagrams~\cite{Joyal1991,Selinger2009} that are used to characterise the arrows of various kinds of free categories. Similar structures have been used by Computer Scientists to develop foundational algebras for composing software components~\cite{Bruni2006,Arbab2004}. Theoretical work has led to tool support for reasoning about different kinds of string diagrams~\cite{Kissinger2012,Sobocinski2013c}.

In~\cite{Soboci'nski2010,Bruni2011,Bruni2013,Sobocinski2013} the author and others have studied compositional algebras of Petri nets. The two main variants, studied in detail in~\cite{Bruni2013}, are \emph{C/E nets with boundaries} and \emph{P/T nets with boundaries}. Nets without places are pure algebras of linkings; we show in this paper that they are, respectively, the arrows of two categories $\Sp{\Relfc}$ and $\Spr{\Relfm}$\footnote{The notation $\Sp{-}$ means ``not quite the category of spans,'' as the objects are the natural numbers, instead of arbitrary sets. Similarly $\Spr{-}$ is ``not quite the category of relational spans,'' where relational means that the two legs are jointly mono. Both categories are PROPs~\cite{MacLane1965,Lack2004a}.}. Recently, string diagrams and closely related algebraic structures have also been used to reason about quantum computation~\cite{Abramsky2004,Selinger2007,Coecke2009}. 

Both categories are generated from a set of basic components, which are the building blocks of \emph{two different} monoid-comonoid structures on the underlying categories. The two structures arise, roughly, from the elementary setting of cospans and spans of finite sets.

In an effort to capture several different kinds of linking algebras, Hughes~\cite{Hughes2008} introduced the category $\mathsf{Link}$ of spans over $\mathbf{iRel}$ the category of injective relations, which has pullbacks. Pullbacks are obtained by considering paths, called \emph{minimal synchronisations}, in the corresponding linking diagrams. Similar ideas are used here in order to construct pullbacks in $\Relfc$, the category of \emph{relations with contention} and weak pullbacks in $\Relfm$, the category of multirelations.
In this paper we study only finite linkings but the category of spans of relations with contention is more expressive than the category of spans of injective relations: the finite counterpart of Hughes' category $\mathsf{Link}$ embeds into $\Sp{\Relfc}$. 

\paragraph{Structure of the paper.}
In \S\ref{sec:components} we introduce the two monoid-comonoid structures that arise from considering cospans and spans of finite sets. In \S\ref{sec:contention} we introduce sets and relations with contention, and show that the category of the latter has pullbacks. This allows us, in \S\ref{sec:linkingdiagscontention} to consider the category $\Sp{\Relfc}$, a universe where both the monoid-comonoid structures can be considered. In \S\ref{sec:multisets} we discuss multirelations and construct weak pullbacks, which we then use in \S\ref{sec:multilinkingdiags} to consider another universe where both the monoid-comonoid structures exist and interact. 

\paragraph{Notational conventions.}
Relations from $X$ to $Y$ are identified with functions $X\to 2^Y$. 
For $k\in\N$ we abuse notation and denote the $k$th finite ordinal $\{0,1\dots,k-1\}$ with $k$. For sets $X$, $Y$, $X+Y\Defeq \{\,(x,0)\,|\,x\in X\,\}\cup\{\,(y,1)\,|\,y\in Y\,\}$. Functions are labelled with $!$ when there is a unique function with that particular domain and codomain, 
$\tw:2\to 2$ is the function $\tw(0)=1$ and $\tw(1)=0$.
Given a function $f:X\to Y$, $[f]\subseteq X\times Y$ is its graph: $[f]\Defeq\{\,(x,fx)\,|\,x\in X\,\}$.
Given a relation $R\subseteq X\times Y$, $R^\op \subseteq Y\times X$ is the opposite relation. 


\section{Components of linking diagrams}
\label{sec:components}


Let $\Csp{\Setf}$ be the category\footnote{Not quite the category of cospans. Again, this is a PROP.} with objects the natural numbers, and arrows
isomorphism classes of cospans $k \rightarrow x \leftarrow l$, where $k$ and $l$ are considered as finite ordinals. Composition is obtained via pushout in $\Setf$, associativity follows from the universal property.
Given $k_1 \rightarrow m_1 \leftarrow l_l$ and $k_2 \rightarrow m_2 \leftarrow l_2$, the tensor product is $k_1 + k_2 \rightarrow m_1+m_2 \leftarrow l_1 + l_2$.

The following diagrams represent certain arrows in $\Csp{\Setf}$.
\begin{figure}
\begin{equation}\label{eq:csp}
\tag{$\diagsym\,\rightEndsym\,\nabla\,\leftEndsym$}
\lower20pt\hbox{$\includegraphics[width=1.5cm]{diag}$}
\qquad
\lower10pt\hbox{$\includegraphics[width=1.5cm]{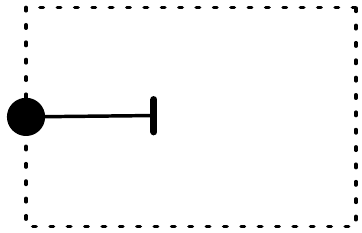}$}
\qquad
\lower20pt\hbox{$\includegraphics[width=1.5cm]{codiag}$}
\qquad
\lower10pt\hbox{$\includegraphics[width=1.5cm]{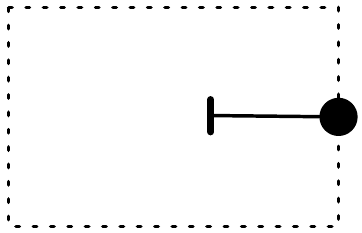}$}
\end{equation}
\[
\diagsym:1\to 2  \qquad
\rightEndsym:1 \to 0 \qquad \quad
\codiagsym:2 \to 1 \qquad
\leftEndsym:0 \to 1\qquad\qquad
\]
\end{figure}
They have representatives
$1 \xrightarrow{\id} 1 \xleftarrow{!} 2$,
$1\xrightarrow{\id} 1 \xleftarrow{!} 0$,
$2 \xrightarrow{!} 1 \xleftarrow{\id} 1$ and 
$0 \xrightarrow{!} 1 \xleftarrow{\id} 1$. 

Our graphical notation calls for further explanation: within the diagrams, each \emph{link}--an undirected multiedge--represents an element of the carrier set, its connections to \emph{boundary ports} (elements of the ordinals on the boundary) are determined in $\Csp{\Setf}$ by the functions from the ordinals that represent the boundaries. Each link has a small perpendicular mark; this is used to distinguish between different links within diagrams.
 
The definition of $\Csp{\Setf}$ enforces some structural restrictions on links. Indeed, each boundary port must be connected to exactly one link; ie no two links can be connected to the same boundary port. 
Any link, however, can be connected to several ports on each boundary. 

\medskip
Now consider $\Sp{\Setf}$, the category with objects the natural numbers, and arrows isomorphism classes of spans $k \leftarrow x \rightarrow l$, where $k$ and $l$ are considered as finite ordinals. Composition is obtained via pullback in $\Setf$, and associativity is again guaranteed by a universal property, this time of pullbacks. Again, $+$ gives a tensor product.

The following diagrams represent certain arrows in $\Sp{\Setf}$.
\begin{figure}
\begin{equation}\label{eq:sp}\tag{$\ldiagsym\!\rzerosym\!\lcodiagsym\!\lzerosym$}
\lower20pt\hbox{$\includegraphics[width=1.5cm]{ldiag}$}
\qquad
\lower10pt\hbox{$\includegraphics[width=1.5cm]{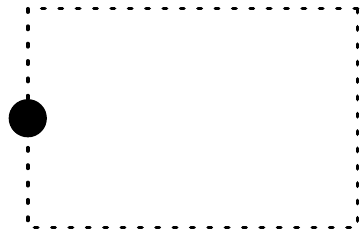}$}
\qquad
\lower20pt\hbox{$\includegraphics[width=1.5cm]{lcodiag}$}
\qquad
\lower10pt\hbox{$\includegraphics[width=1.5cm]{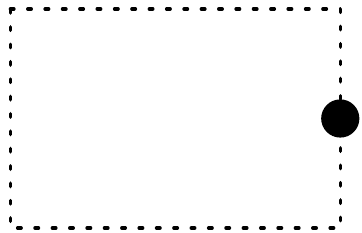}$}
\end{equation}
\[
\ldiagsym:1\to 2  \qquad
\rzerosym:1 \to 0 \qquad
\lcodiagsym:2 \to 1 \qquad
\lzerosym:0 \to 1 \qquad \qquad
\]
\end{figure}
They have representatives $1\xleftarrow{!} 2 \xrightarrow{\id} 2$,
$1\xleftarrow{!} 0 \xrightarrow{\id} 0$, $2\xleftarrow{\id} 2 \xrightarrow{!} 1$ and
$0 \xleftarrow{\id} 0 \xrightarrow{!} 1$.

In the diagrams, the links again represent elements of the carrier set but connections to boundary ports are now given by the functions \emph{from} the carrier \emph{to} the boundaries. Due to the definition of $\Sp{\Setf}$, there are again structural restrictions: each link is connected to exactly one port on each boundary. Any port, however, can be connected to many links.

The following diagrams represent certain arrows in $\Csp{\Setf}$
and $\Sp{\Setf}$. 
\begin{figure}
\begin{equation}\label{eq:common}\tag{$\idsym\,\twsym$}
\lower10pt\hbox{$\includegraphics[width=1.5cm]{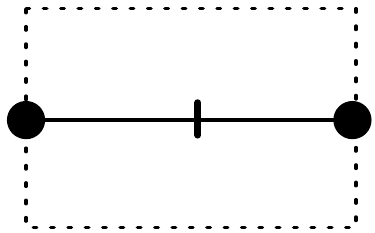}$}
\qquad
\lower19pt\hbox{$\includegraphics[width=1.5cm]{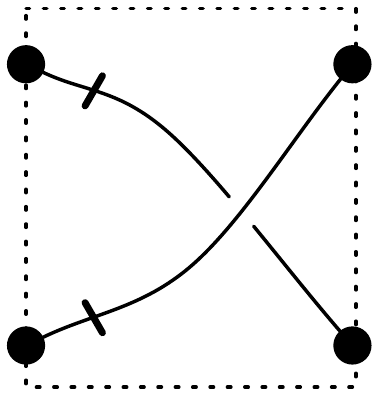}$}
\end{equation}
\[
\idsym: 1\to 1 \qquad \ \twsym: 2\to 2 
\]
\end{figure}
As (isomorphism classes of) cospans they are
$1\rightarrow 1 \leftarrow 1$, $2 \xrightarrow{tw} 2 \leftarrow 2$,
as spans they are $1 \leftarrow 1 \rightarrow 1$, $2 \leftarrow 2 \xrightarrow{tw} 2$.

\subsection{The algebra of $\Csp{\Setf}$}

\begin{figure}
%
%
\begin{equation}\label{eq:cspcomonoid}\tag{$\diagsym$UC}
\begin{tabular}{ c c c c c }
\lower17pt\hbox{$\includegraphics[width=1.5cm]{diag}$} & ; & \lower17pt\hbox{$\includegraphics[width=1.5cm]{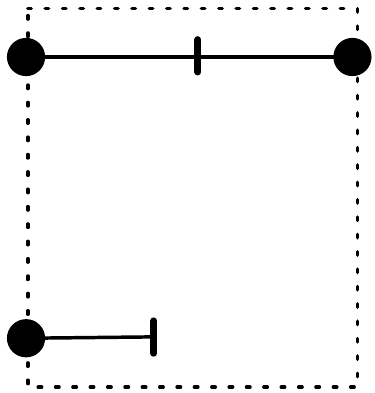}$} & = & \lower9pt\hbox{$\includegraphics[width=1.5cm]{id}$} \\
\end{tabular}
\ 
\begin{tabular}{ c c c c c }
\lower17pt\hbox{$\includegraphics[width=1.5cm]{diag}$} & ; & \lower17pt\hbox{$\includegraphics[width=1.5cm]{tw}$} & = &
\lower17pt\hbox{$\includegraphics[width=1.5cm]{diag}$} \\
\end{tabular}
\end{equation}
%
%
\begin{equation}\label{eq:cspassoc}\tag{$\diagsym$A}
\lower18pt\hbox{$\includegraphics[width=1.5cm]{diag}$} \ \comp \ 
\lower30pt\hbox{$\includegraphics[width=1.5cm]{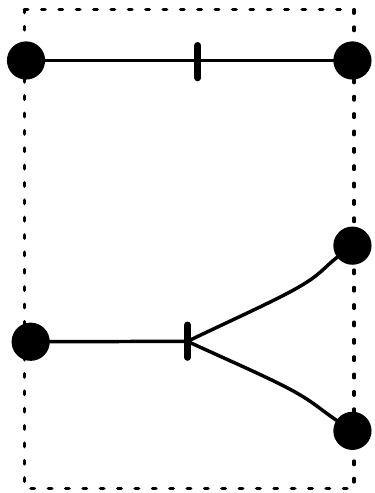}$}
\quad = \quad
\lower28pt\hbox{$\includegraphics[width=1.2cm]{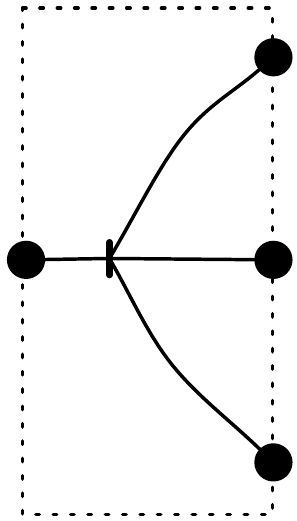}$}
\quad = \quad
\lower20pt\hbox{$\includegraphics[width=1.5cm]{diag}$} \ \comp \ 
\lower22pt\hbox{$\includegraphics[width=1.5cm]{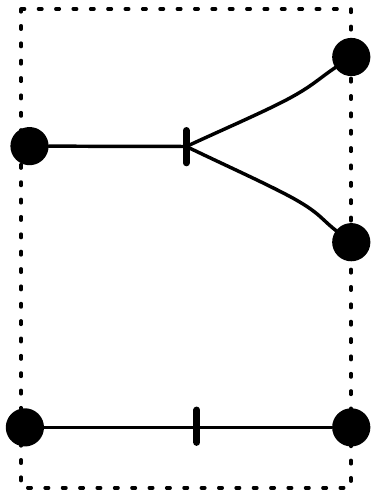}$}
\end{equation}
%
%
\begin{equation}\label{eq:frobenius}\tag{F}
\lower20pt\hbox{$\includegraphics[width=1.2cm]{idtensordiag}$}
\ \comp \ 
\lower20pt\hbox{$\includegraphics[width=1.2cm]{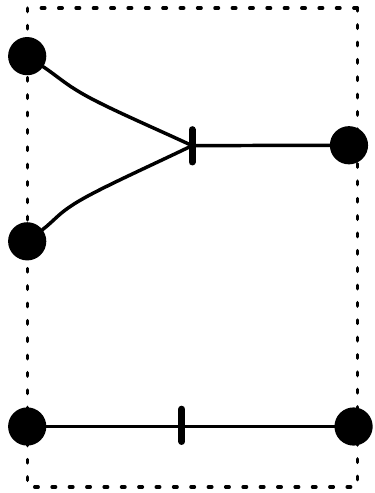}$}
\ = \ 
\lower20pt\hbox{$\includegraphics[width=1.2cm]{diagtensorid}$}
\ \comp \ 
\lower20pt\hbox{$\includegraphics[width=1.2cm]{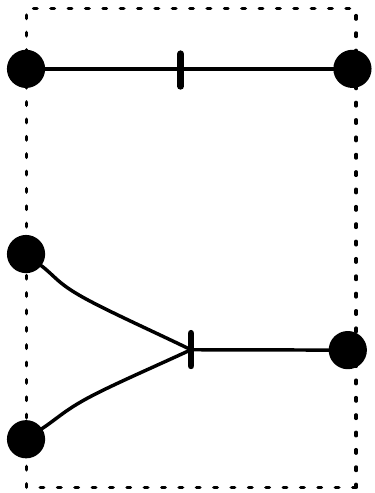}$}
\ = \ 
\lower15pt\hbox{$\includegraphics[width=1.2cm]{codiag}$} 
\ \comp \ 
\lower15pt\hbox{$\includegraphics[width=1.2cm]{diag}$} 
\ = \ 
\lower15pt\hbox{$\includegraphics[width=1.2cm]{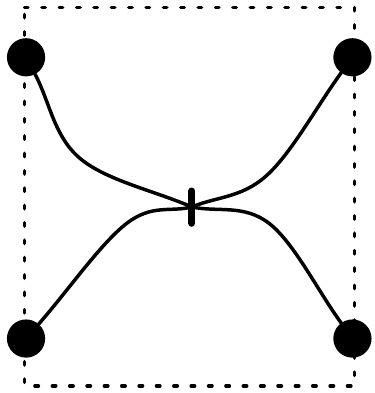}$} 
\end{equation}
%
%
\begin{equation}\label{eq:separable}\tag{S}
\lower19pt\hbox{$\includegraphics[width=1.5cm]{diag}$} \ \comp \ 
\lower19pt\hbox{$\includegraphics[width=1.5cm]{codiag}$} \ = \  \lower11pt\hbox{$\includegraphics[width=1.5cm]{id}$} 
\end{equation}
\begin{equation}\label{eq:compactclosed}\tag{CC}
\lower17pt\hbox{$\includegraphics[width=1.5cm]{codiag}$} \ \comp \ 
\lower9pt\hbox{$\includegraphics[width=1.5cm]{rightend}$} \ = \ 
\lower17pt\hbox{$\includegraphics[width=1.5cm]{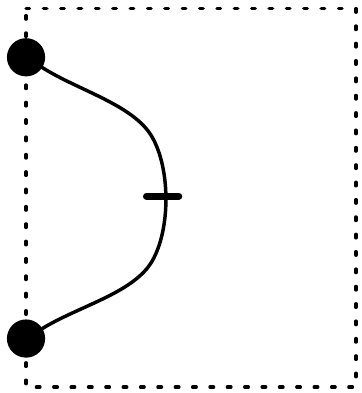}$}
\end{equation}
\caption{Equations in $\Csp{\Setf}$.\label{fig:csp}}
\end{figure}
In \figref{fig:csp} we give some of the equations satisfied by the algebra generated from the components~\eqref{eq:csp} and~\eqref{eq:common} in $\Csp{\Setf}$: \eqref{eq:cspcomonoid} and~\eqref{eq:cspassoc} show that $\diagsym$  is the comultiplication of a cocommutative comonoid. The symmetric equations hold for $\codiagsym$, meaning that it is part of a commutative monoid structure. The Frobenius axioms~\eqref{eq:frobenius}~\cite{Carboni1987,Kock2003} hold, and the algebra is separable~\eqref{eq:separable}. In fact $\Csp{\Setf}$ is the free PROP on \eqref{eq:csp} satisfying such axioms, where~\eqref{eq:frobenius}, \eqref{eq:separable} can be understood as witnessing a distributive law of PROPs; see~\cite{Lack2004a} for the details.
In~\eqref{eq:compactclosed} we indicate how the (self dual) compact closed structure of $\Csp{\Setf}$ arises.


\subsection{The algebra of $\Sp{\Setf}$}


\begin{figure}
\begin{equation}\label{eq:spcomonoid}\tag{$\ldiagsym$UC}
\begin{tabular}{ c c c c c }
\lower19pt\hbox{$\includegraphics[width=1.5cm]{ldiag}$} & ; &
\lower19pt\hbox{$\includegraphics[width=1.5cm]{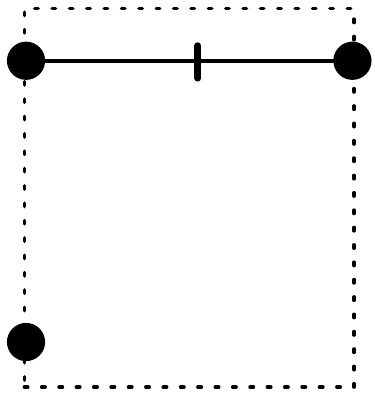}$} & = & \lower10pt\hbox{$\includegraphics[width=1.5cm]{id}$} \\
\end{tabular}
\ 
\begin{tabular}{ c c c c c}
\lower19pt\hbox{$\includegraphics[width=1.5cm]{ldiag}$} & ; &
\lower19pt\hbox{$\includegraphics[width=1.5cm]{tw}$} & = &
\lower19pt\hbox{$\includegraphics[width=1.5cm]{ldiag}$} \\
\end{tabular}
\end{equation}
\begin{equation}\label{eq:spassoc}\tag{$\ldiagsym$A}
\lower18pt\hbox{$\includegraphics[width=1.5cm]{ldiag}$} \ \comp \ 
\lower30pt\hbox{$\includegraphics[width=1.5cm]{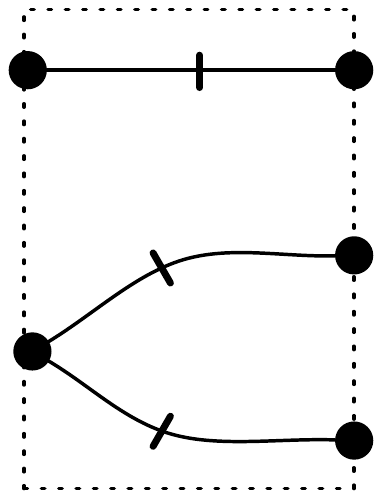}$}
\quad = \quad
\lower30pt\hbox{$\includegraphics[width=1.3cm]{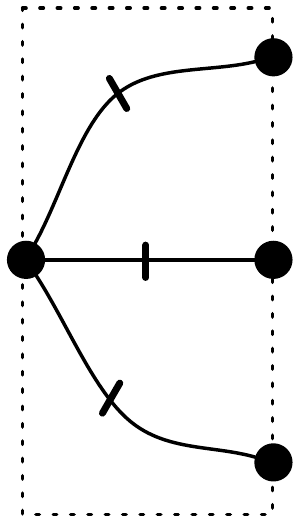}$}
\quad = \quad
\lower20pt\hbox{$\includegraphics[width=1.5cm]{ldiag}$} \ \comp \ 
\lower22pt\hbox{$\includegraphics[width=1.5cm]{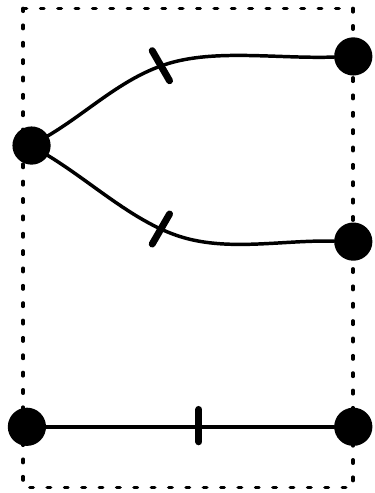}$}
\end{equation}
\begin{equation}\label{eq:bialgebra}\tag{B}
\lower37pt\hbox{$\includegraphics[width=1.4cm]{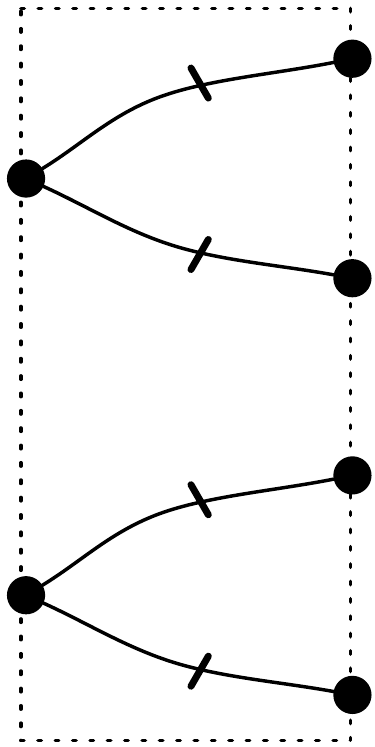}$}
\ \comp\ 
\lower37pt\hbox{$\includegraphics[width=1.4cm]{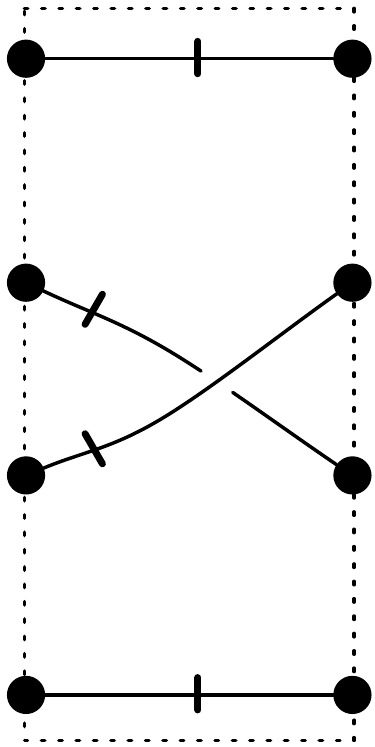}$}
\ \comp\ 
\lower37pt\hbox{$\includegraphics[width=1.4cm]{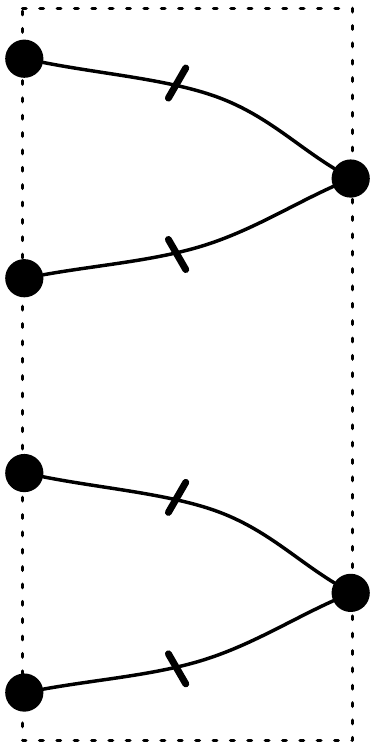}$}
\  = \ 
\lower20pt\hbox{$\includegraphics[width=1.5cm]{lcodiag}$} 
\comp 
\lower20pt\hbox{$\includegraphics[width=1.5cm]{ldiag}$}
\  = \ 
\lower20pt\hbox{$\includegraphics[width=1.5cm]{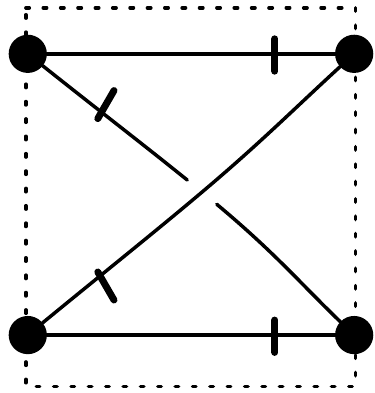}$} 
\end{equation}
\begin{equation}\label{eq:reverseunit}\tag{$\lcodiagsym\!\!\rzerosym$}
\lower20pt\hbox{$\includegraphics[width=1.5cm]{lcodiag}$} \comp
\lower10pt\hbox{$\includegraphics[width=1.5cm]{rzero}$}  = 
\lower20pt\hbox{$\includegraphics[width=1.5cm]{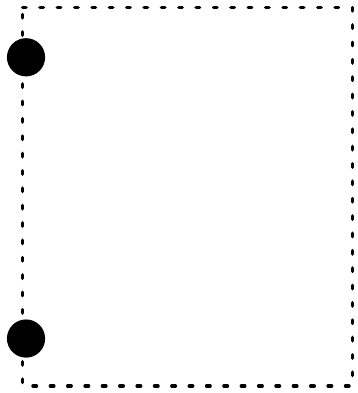}$}
\end{equation}
\begin{equation}\label{eq:diagcodiag}\tag{$\ldiagsym\lcodiagsym$}
\lower20pt\hbox{$\includegraphics[width=1.5cm]{ldiag}$} \comp
\lower20pt\hbox{$\includegraphics[width=1.5cm]{lcodiag}$} =
\lower13pt\hbox{$\includegraphics[width=1.5cm]{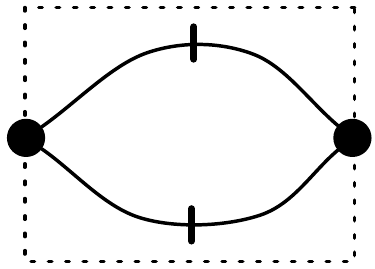}$} 
\end{equation}
\caption{Equations in $\Sp{\Setf}$.\label{fig:sp}}
\end{figure}
In \figref{fig:sp} we exhibit some equations satisfied by the components~\eqref{eq:sp} and~\eqref{eq:common} in $\Sp{\Setf}$: \eqref{eq:spcomonoid} and~\eqref{eq:spassoc} show that $\ldiagsym$ is the multiplication of a cocommutative comonoid, similarly the symmetric equations, which we do not illustrate, show that that $\lcodiagsym$ is a commutative monoid. Differently from \figref{fig:csp}, here the Frobenius equations do not hold; but rather the equations of commutative and cocommutative bialgebras: in \eqref{eq:bialgebra}, \eqref{eq:reverseunit} and \eqref{eq:diagcodiag} we show how the monoid and comonoid structures interact in $\Sp{\Setf}$. In fact, $\Sp{\Setf}$ is the free PROP on \eqref{eq:sp} satisfying the equations of commutative and cocommutative bialgebras, and the bialgebra axiom can be understood as a distributive law of PROPs, see~\cite{Lack2004a}.

\subsection{Bringing it all together}
Note that none of the diagrams in \eqref{eq:csp} represent valid spans: for instance the link in $\diagsym$ connects to two different ports on its right boundary, and the link in $\rightEndsym$ does not connect to any port on its right boundary. Similarly, none of \eqref{eq:sp} represent valid cospans. Thus, for mundane ``expressivity'' reasons, \eqref{eq:csp} are not arrows of $\Sp{\Setf}$, and vice-versa, \eqref{eq:sp} are not arrows of $\Csp{\Setf}$.
The remit of this paper is to study how these two commutative monoid-comonoid structures 
interact together in universes that are expressive enough to accommodate them.

For example, instead of studying cospans and spans of \emph{functions}, one could consider spans (or cospans) of \emph{relations}. Indeed, it is not difficult to check that all of the components~\eqref{eq:csp}, \eqref{eq:sp} and~\eqref{eq:common} are spans of relations of finite sets. The problem, of course, is that $\Rel_f$, the category of finite sets and relations, does not have pullbacks nor pushouts: it is thus not clear how to define the composition of such linking diagrams.

In the following sections we study two different universes that are expressive enough to contain~\eqref{eq:csp}, \eqref{eq:sp} and \eqref{eq:common} and the intriguing, different ways in which the two monoid/comonoid structures interact in the universes. They arose through the study of compositional algebras of Petri nets with boundaries~\cite{Soboci'nski2010,Bruni2011,Bruni2013,Sobocinski2013}.


\section{Sets with contention}
\label{sec:contention}


In this section we introduce \emph{sets with contention}, over which one can define a category of relations that has pullbacks, and is expressive enough to accommodate the components~\eqref{eq:csp}, \eqref{eq:sp} and~\eqref{eq:common}. 
\smallskip

A \emph{set with contention}, or $c$-set, is a pair $(X, \contention{X})$, where $X$ is a set and $\contention{X}\subseteq X\times X$ is a reflexive $(\forall x\in X.\; (x \contention{X} x))$ and symmetric $(\forall x,y\in X.\; (x \contention{X} y) \Rightarrow (y \contention{X} x))$ relation called \emph{contention}.\footnote{A useful intuition is that links carry signals. When two links are in contention they cannot transmit concurrently. With this intuition \eqref{eq:csp} are copy and forget operations, while 
\eqref{eq:sp} are non-deterministic switches and ``failure.''}
 The complement relation $\ind{X}$ is called \emph{independence}. To describe a $c$-set it is thus of course enough to specify either contention or independence. Sets with contention of the form $(X, \delta_X)$, where $\delta_X=\{\,(x,x)\;|\;x\in X\,\}$,  are said to be \emph{discrete}. We will normally write simply $X$ for the pair $(X,\contention{X})$.

A \emph{morphism of $c$-sets} $f:X\to Y$ is a function $f:X\to Y$ such that: 
\begin{equation}\label{eq:morphism}
\forall x,x'\in X.\; f(x)\contention{Y} f(x') \Rightarrow x \contention{X} x'
\end{equation}
 (or equivalently $\forall x,x'\in X$, $x \ind{X} x'$ implies $fx 
\ind{Y} fx'$.) The category of finite $c$-sets and their morphisms is denoted $\Setfc$.  

Given $c$-sets $X_0$ and $X_1$, $X_0+X_1$ is the $c$-set with $X_0+X_1$ as its underlying set and $(x,i)\contention{X_0+X_1}(y,j)$ iff $i=j$ and $x \contention{X_i} y$. This is the categorical coproduct in $\Setfc$.

Given a $c$-set $X$, $U\subseteq X$ is said to be \emph{independent} when 
\begin{equation}\label{eq:independent}
\forall u, u'\in U.\; u\contention{X} u'\ \Rightarrow\ u = u'.
\end{equation}
Let $\Pc X$ denote the set of independent subsets of $X$. 
There is functor $\Pc:\Setfc\to \Setfc$ that takes a $c$-set $X$ to the set of independent subsets $\Pc X$, with contention between subsets defined:
\[ U \contention{\Pc X} V  \text{ iff } \exists u\in U, v\in V, u \contention{X} v.\]
Note that independent subsets are closed under intersection and set difference: indeed, if $U'\subseteq U$ and $U$ is independent then also $U'$ is independent. They are not, in general, closed under union.

If $f:X\to Y$ is a morphism, then letting
\[ \Pc f (U) \Defeq \{\, fu \;|\; u\in U\,\} \]
defines a morphism $\Pc f : \Pc X \to \Pc Y$ in $\Setfc$, since: 
\begin{enumerate}[(i)]
\item given $U$, for all $u, u'\in U$ if $f(u)\contention{Y}f(u')$ means that $u\contention{X} u'$. But $U$ is independent, and thus $u=u'$ and $f(u)=f(u')$, thus $\Pc f (U)$ is an independent subset of $Y$ (recall~\eqref{eq:independent}).
\item if $\Pc f (U) \contention{Y} \Pc f (V)$ then there exists $u\in U$, $v\in V$, such that $f(u)\contention{Y} f(v)$, so $u\contention{X} v$ and thus $U\contention{X} V$, thus
$\Pc f$ satisfies~\eqref{eq:morphism}.
\end{enumerate}


\subsection{Relations with contention}


There are morphisms
$\mu_X : \Pc^2 X \to \Pc X$ with
$\{U_i\} \mapsto \bigcup_i U_i$
and a morphism $\eta_X: X \to \Pc X$. It
is not difficult to check that they are natural transformations that satisfy the monad axioms.

Let $\Relfc \Defeq \Kl{\Pc}$ of relations with contention, or $c$-relations, be the Kleisli category with objects finite $c$-sets. Arrows from $X$ to $Y$ are morphisms $f: X \to \Pc Y$ in $\Setfc$, which we will sometimes denote $f: X \relto Y$. Given a morphism $f: X \to \Pc Y$ in $\Setfc$ (or equivalently, a morphism of $\Relfc$), $\lift{f}: \Pc X\to \Pc Y$ is the morphism $\lift{f}U\Defeq \bigcup_{u\in U}fu $.

The following lemma is useful when calculating in $\Relfc$. It does not hold in $\Rel_f$, the category of ordinary finite sets and relations.
\begin{lem}\label{lem:difference}
Suppose $f:X \to \Pc Y$ in $\Setfc$.
Then, given $U,U'\in \Pc X$ with $U\subseteq U'$, $\lift{f}(U'\backslash U)=\lift{f}(U')\backslash\lift{f}(U)$. Also, given $U, V, V'\in \Pc X$, with $V\subseteq U$, $V'\subseteq U$, we have $\lift{f}(V\cap V')=\lift{f}(V)\cap \lift{f}(V')$.
\end{lem}
\begin{proof}
Since $U'$ is independent, $\{f u \}_{u\in U'}$ is a family of
disjoint, independent subsets of $Y$. Similarly $V\cup V'$ is independent, 
since they are both subsets of an independent set; and $\{f u\}_{u\in V\cup V'}$
is a family of disjoint, independent subsets of $Y$.
Disjointness implies the desired conclusions. \qed
\end{proof}


\subsection{Pullbacks in $\Relfc$}
\label{subsec:synchronisations}


Suppose that $f: A \relto X$ and 
$g: B \relto X$ in $\Relfc$.
Given $U\in \Pc A$, $V\in  \Pc B$, say that $(U,V)$ is a ($f$,$g$)-\emph{synchronisation}\footnote{Hughes~\cite{Hughes2008} uses the term synchronisation in a similar context, and the term has been used in~\cite{Soboci'nski2010,Bruni2011,Bruni2013} to compose Petri nets with boundaries.} if $\lift{f}U = \lift{g}V$. We will typically infer $f$ and $g$ from the context and write `$\synch{U}{V}$' as shorthand for `a synchronisation $(U,V)$'. 
Synchronisations inherit an ordering from the subset ordering, pointwise: \[
\synch{U}{V}\subseteq \synch{U'}{V'} \ \Defeq\  U\subseteq U' \wedge V\subseteq V'.\]
 The \emph{trivial} synchronisation is $\synch{\varnothing}{\varnothing}$. A synchronisation $\synch{U}{V}$ is said to be \emph{minimal} when it is not trivial and for all $\synch{U'
}{V'}$ such that $\synch{U'}{V'}\subseteq \synch{U}{V}$, either $\synch{U'}{V'}$ is trivial or equal to $\synch{U}{V}$.

Let $\minsynch{f}{g}$ be the set of minimal synchronisations of $f$ and $g$. We can define contention on this set by letting 
\[
\synch{U}{V} \contention{\minsynch{f}{g}} \synch{U'}{V'} \quad\Defeq\quad  U\contention{\P A} U' \ \vee\  V\contention{\P B} V'.
\]
It follows that have the following commutative diagram in $\Relfc$
\begin{equation}\label{eq:pb}
\raise3.5em\hbox{$
\xymatrix@=15pt{
 & \minsynch{f}{g} \ar[dl]_-p \ar[dr]^-q \\
 A \ar[dr]_f & & B  \ar[dl]^g \\
 & X 
}$}
\end{equation}
where $p\synch{U}{V}=U$ and $q\synch{U}{V}=V$. The following observations will lead us to conclude in Lemma~\ref{lem:pb} that the diagram is a pullback in $\Relfc$. 


Synchronisations are not in general closed under (pointwise) union, because if $\synch{U}{V}$ and $\synch{U'}{V'}$ then in general it is not true that $U\cup U'\in \Pc A$ and $V\cup V' \in \Pc B$. It is true, however, that the union of any set of minimal synchronisations contained in any synchronisation is again a synchronisation: this is guaranteed by the following.
\begin{lem}\label{lem:synchintersection}
Suppose that $\synch{U'}{V'}\neq \synch{U''}{V''}$ are 
minimal synchronisations contained in $\synch{U}{V}$. Then $U'\cap U''=\varnothing$
and $V' \cap V'' = \varnothing$.
\end{lem}
\begin{proof}
By the conclusion of Lemma~\ref{lem:difference}, 
$\lift{f}(U\cap U') = \lift{f}U \cap \lift{f} U' = \lift{g} V \cap \lift{g} V' = \lift{g} (V \cap V')$, so $\synch{U'\cap U''}{V'\cap V''}$; by minimality of $\synch{U'}{V'}$ and $\synch{U''}{V''}$ it follows that $\synch{U'\cap U''}{V'\cap V''}$ is trivial. \qed
\end{proof}

\begin{lem}\label{lem:synchdecomposition}
$\synch{U}{V}$ is the union of min. synchronisations it contains.
\end{lem}
\begin{proof}
Let $\{\synch{U_i}{V_i}\}_{i\in I}$ be the set of minimal synchronisations contained in $\synch{U}{V}$ and
$\synch{U'}{V'} \Defeq \bigcup_i \{\synch{U_i}{V_i}\}$, then
clearly we have $\synch{U'}{V'} \subseteq \synch{U}{V}$. 
Let $U''=U\backslash U'$ and  $V''=V\backslash V'$. 
Now, using the conclusion of Lemma~\ref{lem:difference},
$\synch{U''}{V''}$, and thus it is either null or it contains a minimal synchronisation. But $\{\synch{U_i}{V_i}\}_{i\in I}$ contains all minimal synchronisations in $\synch{U}{V}$; thus $U''=V''=\varnothing$ and we are finished. \qed
\end{proof}

\begin{lem}\label{lem:pb}
The square~\eqref{eq:pb} is a pullback diagram in $\Relfc$.
\end{lem}
\begin{proof}
Suppose $Z$ is a $c$-set and $\alpha:Z \to A$, $\beta: Z \to B$ are morphisms in $\Relfc$ such that $f\alpha = g\beta$. In particular, this means that for all $z\in Z$, we have
$\synch{\alpha z}{\beta z}$. Define $h: Z \to \minsynch{f}{g}$ by letting $h z$ be the family of minimal synchronisations contained in $\synch{\alpha z}{\beta z}$. This is a independent set, due to Lemma~\ref{lem:synchintersection}, and the fact that $\alpha z$ and $\beta z$ are independent. Then, by the conclusion of Lemma~\ref{lem:synchdecomposition}, $p h = \alpha$ and $q h = \beta$. 

If another $h'$ satisfies $p h' = \alpha$ and $q h' = \beta$ then there exists a family of minimal synchronisations $h' z = \{\synch{U_i}{V_i}\}_{i\in I}$ such that 
$\bigcup_i U_i = \alpha z$ and $\bigcup_i V_i = \beta z$. By the conclusion of Lemma~\ref{lem:synchintersection} this family must be $h z$. \qed
\end{proof}


\section{The algebra of $\Sp{\Relfc}$}
\label{sec:linkingdiagscontention}


In this section we consider a category with enough structure for all of \eqref{eq:csp}, \eqref{eq:sp} and \eqref{eq:common}. It has been considered as part of a compositional algebra of C/E (1 bounded) nets~\cite{Soboci'nski2010,Bruni2013}---indeed, it is the category of C/E nets with boundaries, without net places, up to isomorphism.
\smallskip

Consider the category $\Sp{\Relfc}$, that has objects the natural numbers
and arrows $k \to l$ isomorphism classes of spans $k \xrelleftarrow{f} (X,\contention{X}) \xrelrightarrow{g} l$ in $\Relfc$, 
where $k$ and $l$ are considered as discrete $c$-sets. 
Composition is via pullback in $\Relfc$; associativity follows from the universal property. There is a tensor product, given by $+$.

$\Sp{\Relfc}$ has enough structure for~\eqref{eq:csp}, \eqref{eq:sp} and~\eqref{eq:common}. Indeed, \eqref{eq:csp} are, respectively, spans $1\xrelleftarrow{[\id]} 1 \xrelrightarrow{[!]^\op} 2$, 
$1\xrelleftarrow{[\id]} 1 \xrelrightarrow{[!]^\op} 0$, $2 \xrelleftarrow{[!]^\op} 1 \xrelrightarrow{[\id]} 1$
and $0\xrelleftarrow{[!]^\op} 1 \xrelrightarrow{[\id]} 1$. Similarly, \eqref{eq:common} are spans $1 \xrelleftarrow{[\id]} 1 \xrelrightarrow{[\id]}$ and $2 \xrelleftarrow{[\id]} 2 \xrelrightarrow{[tw]} 2$.
Indeed, $\Csp{\Setf}$ embeds into $\Sp{\Relfc}$.
\begin{theorem}
There is a faithful functor $E:\Csp{\Setf} \to \Sp{\Relfc}$ that is identity-on-objects.
\end{theorem}
\begin{proof}
A cospan $k \xrightarrow{f} x \xleftarrow{g} l$ is
taken to $k \xrelleftarrow{[f]^\op} x \xrelrightarrow{[g]^\op} l$, where
$k$, $x$ and $l$ are discrete $c$-sets, and
$[f]^\op$, $[g]^\op$ are the opposites of graphs of, respectively, $f$ and $g$. 
As arrows in $\Kl{\Pc}$, 
$[f]^\op u={f}^{-1}u$ and $[g]^\op u={g}^{-1}u$ for any $u\in x$.
Identities are clearly preserved. 

We must show that composition is preserved; it suffices
to show that, given $g_0:l \to x_0$ and $f_1: l \to x_1$, a pushout diagram of $g_0$, $f_1$ in $\Setf$ is taken to a pullback diagram in $\Relfc$, as illustrated below.
\begin{equation}\label{eq:embeddingcsp}
\raise25pt\hbox{$
\xymatrix@=15pt{
& l \ar[dl]_{g_0} \ar[dr]^{f_1} \\
x_0 \ar[dr]_-r & & x_1 \ar[dl]^-s \\
& 
\ar@{}[uu]|(.2){\hbox{\rotatebox[origin=c]{270}{$\langle$}}}
M
}$}
\ 
\longmapsto
\ 
\raise25pt\hbox{$\xymatrix@=15pt{
& {M} \ar@{}[dd]|(.2){\hbox{\rotatebox[origin=c]{90}{$\langle$}}}
\ar[dl]_-{[r]^\op} \ar[dr]^-{[s]^\op}
\\
x_0 \ar[dr]_{[g_0]^\op} & & x_1 \ar[dl]^{[f_1]^\op} \\
 & l 
}$}
\end{equation}
If $M=0$ then also $x_0=x_1=l=0$ and all arrows are $\id_0$.
Otherwise, by an inductive argument it suffices to consider the case $M=1$. 
In that case, if $x_0=0$ then $x_1=1$ and $l=0$.
Then $\minsynch{[g_0]^\op}{[f_1]^\op}=\{\synch{\varnothing}{1}\}$ and
we are done. The case $x_1=0$ is symmetric. If both $x_0,x_1\neq 0$ then clearly
$\synch{x_0}{x_1}$. In fact, it is the only non-trivial synchronisation (and thus minimal). To see this, notice that $g_0$ and $f_1$ are surjective and therefore, if $\synch{U_1}{V_1}$ and $\synch{U_2}{V_2}$ are two different non-trivial synchronisations then $l_1\Defeq g_0^{-1}U_1=f_1^{-1}V_1\neq \varnothing$ and 
$l_2\Defeq g_0^{-1}U_2=f_1^{-1}V_2\neq \varnothing$, but $l_1\cap l_2=\varnothing$. This means that $l=l_1+l_2+l_3$, for some $l_3$, and the whole left hand side of~\eqref{eq:embeddingcsp} decomposes into a sum, 
contradicting the assumption that $M=1$.

The inductive argument relies on sums being compatible with pullbacks in $\Relfc$. This follows from the construction: minimal synchronisations of $x_0+x_0'\xrelrightarrow{[g_0+g_0']^{\op}} l+l'\xrelleftarrow{[f_0+f_1']^{\op}} x_1+x_1'$ arise either as a minimal synchronisations of $[g_0]^\op$ and $[f_1]^\op$,
or those of $[g_0']^\op$ and $[f_1']^\op$. 
\qed
\end{proof}
As a consequence, the equations for~\eqref{eq:csp} ---presented in~\eqref{eq:cspcomonoid}, \eqref{eq:cspassoc}, \eqref{eq:frobenius}, \eqref{eq:separable} and \eqref{eq:compactclosed}--- also hold in $\Sp{\Relfc}$.

Also \eqref{eq:sp} are spans of $c$-relations: 
$1 \xrelleftarrow{[!]} (2,2\t 2) \xrelrightarrow{[\id]} 2$, $1 \xrelleftarrow{[!]} 0 \xrelrightarrow{[\id]} 0$, $2 \xrelleftarrow{[\id]} (2,2\t 2) \xrelrightarrow{[!]} 1$ and $0 \xrelleftarrow{[\id]} 0 \xrelrightarrow{[!]} 1$; notice that contention is used to ``encode''~\eqref{eq:sp}. This is necessary because the two elements of $2$ must be in contention in order for $!:2\to 1$ to be a $c$-morphism. 
\begin{remark}
When considering, for instance $\ldiagsym$ of~\eqref{eq:sp} we are in a situation where two links connect to the same point on the boundary. Since any element is in contention with itself, this means that the two links must be in contention. Thus, in this example, contention between the two links is implied and we will not alter our graphical notation. We will, however, need a way to represent contention graphically when it is not implied ``structurally,'' and we will do this by connecting the links with dotted lines. For instance, the two diagrams below represent the spans $2 \xleftarrow{[\id]} (2,2\t 2) \xrightarrow{[\id]} 2$ and $2 \xleftarrow{[\id]} (2, 2\t 2) \xrightarrow{[\tw]} 2$.
\begin{equation}\label{eq:idandtwistcontention}
\lower19pt\hbox{$\includegraphics[width=1.5cm]{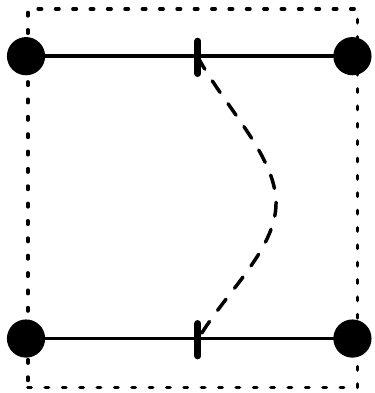}$} 
\qquad\qquad
\lower19pt\hbox{$\includegraphics[width=1.5cm]{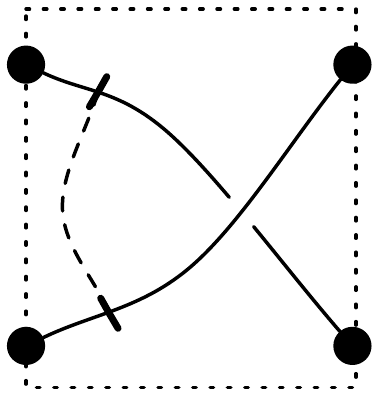}$}
\end{equation}
\end{remark}

\begin{remark}
There is also an ``embedding'' $F: \Sp{\Setf} \to \Sp{\Relfc}$. A span 
$k \xleftarrow{f} x \xrightarrow{g} l$ is sent to the span
$k \xrelleftarrow{[f]} (x,x\times x) \xrelrightarrow{[g]} l$,
with the carrier set having all elements in contention. It is not difficult to check that composition is preserved, but the mapping fails to be a functor because identities are not preserved. For instance, the identity on $2$ is mapped to the left diagram of~\eqref{eq:idandtwistcontention}, which is not the identity on $2$ in $\Sp{\Relfc}$.
\end{remark}

The finite fragment of Hughes' category $\mathsf{Link}$ of spans of \emph{injective} relations~\cite{Hughes2008} lies between $\Csp{\Set_f}$ and $\Sp{\Relfc}$. Indeed, spans of injective relations are expressive enough to consider all the structure of~\eqref{eq:csp}, \eqref{eq:common} and the units
$\rzerosym$, $\lzerosym$ of \eqref{eq:sp}; but not the comultiplication and multiplication $\ldiagsym$, $\lcodiagsym$ --- these are not injective relations. $\mathsf{Link}$ embeds into $\Sp{\Relfc}$, thus all the equations that hold in the former hold also in the latter. We omit the details here.

Equations~\eqref{eq:spcomonoid}, \eqref{eq:spassoc}, \eqref{eq:reverseunit} and \eqref{eq:diagcodiag} hold in in $\Sp{\Relfc}$. Equation~\eqref{eq:bialgebra} does not hold: while we have 
\begin{equation}\label{eq:fourwires}\tag{$B^c$}
\lower35pt\hbox{$\includegraphics[width=1.3cm]{ldiagtensorldiag}$}
\ \comp\ 
\lower35pt\hbox{$\includegraphics[width=1.3cm]{idtensortwtensorid}$}
\ \comp\ 
\lower35pt\hbox{$\includegraphics[width=1.3cm]{lcodiagtensorlcodiag}$}
\  = \ 
\lower18pt\hbox{$\includegraphics[width=1.4cm]{full}$} 
\end{equation}
we have
\begin{equation}\label{eq:onewire}\tag{$\lcodiagsym\ldiagsym^c$}
\lower18pt\hbox{$\includegraphics[width=1.4cm]{lcodiag}$} 
\comp 
\lower18pt\hbox{$\includegraphics[width=1.4cm]{ldiag}$}
\ = \ 
\lower18pt\hbox{$\includegraphics[width=1.4cm]{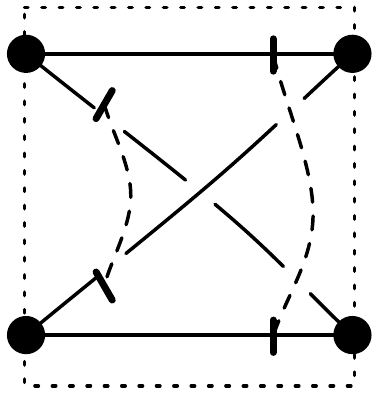}$}
\end{equation}
and the right-hand sides are not equal as arrows of $\Sp{\Relfc}$.


In \eqref{eq:diagandzero}, \eqref{eq:ldiagandend}, \eqref{eq:diagandlcodiag}
and \eqref{eq:diagandldiag} below we show how \eqref{eq:csp} and \eqref{eq:sp} interact together in $\Sp{\Relfc}$. We comment on two of the more interesting equations that the interactions suggest: the right hand side of \eqref{eq:diagandlcodiag} implies 
$\lcodiagsym\comp\diagsym=(\diagsym\otimes\diagsym)\comp(\idsym\otimes\twsym\otimes\idsym)\comp(\lcodiagsym\otimes\lcodiagsym)$,  an ``asymmetric'' commutative/cocommutative bialgebra structure. The left hand side of \eqref{eq:diagandldiag} implies $\diagsym \comp (\ldiagsym \otimes \idsym) = \ldiagsym \comp (\diagsym\otimes\diagsym)\comp (\idsym\otimes \twsym\otimes\idsym) \comp (\idsym\otimes\idsym\otimes \lcodiagsym)$.
 
\begin{equation}\label{eq:diagandzero}\tag{$\diagsym\!\!\rzerosym\!\!\lzerosym\!\!{}^c$}
\lower18pt\hbox{$\includegraphics[width=1.2cm]{diag}$} \ \comp\ 
\lower18pt\hbox{$\includegraphics[width=1.2cm]{idtensorrzero}$} \ =\   \lower9pt\hbox{$\includegraphics[width=1.2cm]{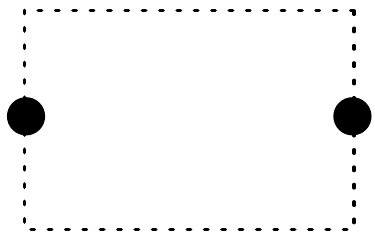}$}
\quad
\lower18pt\hbox{$\includegraphics[width=1.2cm]{codiag}$} \ \comp\  \lower9pt\hbox{$\includegraphics[width=1.2cm]{rzero}$} \ =\  \lower18pt\hbox{$\includegraphics[width=1.2cm]{rzerotensorrzero}$}
\end{equation}
\begin{equation}\label{eq:ldiagandend}\tag{$\ldiagsym\leftEndsym\rightEndsym^c$}
\lower18pt\hbox{$\includegraphics[width=1.2cm]{ldiag}$} \ \comp\  \lower18pt\hbox{$\includegraphics[width=1.2cm]{idtensorrightend}$} \ =\  \lower9pt\hbox{$\includegraphics[width=1.2cm]{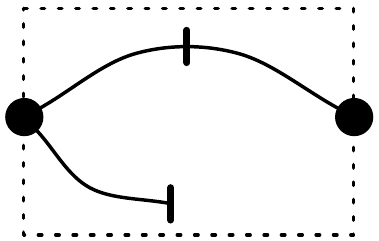}$} 
\quad
\lower18pt\hbox{$\includegraphics[width=1.2cm]{lcodiag}$} \ \comp\  \lower9pt\hbox{$\includegraphics[width=1.2cm]{rightend}$} \ =\  \lower18pt\hbox{$\includegraphics[width=1.2cm]{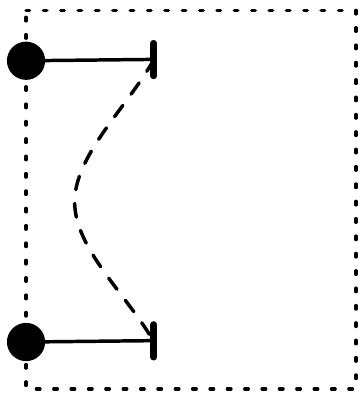}$}
\end{equation}
\begin{equation}\label{eq:diagandlcodiag} \tag{$\diagsym\lcodiagsym^c$}
\lower18pt\hbox{$\includegraphics[width=1.2cm]{diag}$} \ \comp\ 
\lower18pt\hbox{$\includegraphics[width=1.2cm]{lcodiag}$} \ =\  \lower9pt\hbox{$\includegraphics[width=1.2cm]{rzerocomplzero}$}
\quad
\lower18pt\hbox{$\includegraphics[width=1.2cm]{lcodiag}$} \ \comp\ 
\lower18pt\hbox{$\includegraphics[width=1.2cm]{diag}$} \ =\ 
\lower18pt\hbox{$\includegraphics[width=1.2cm]{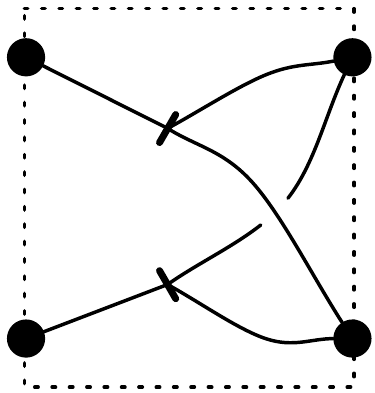}$}
\end{equation}
\begin{equation}\label{eq:diagandldiag} \tag{$\diagsym\ldiagsym^c$}
\lower18pt\hbox{$\includegraphics[width=1.2cm]{diag}$} \ \comp\ 
\lower22pt\hbox{$\includegraphics[width=1.3cm]{ldiagtensorid}$} \ =\  \lower22pt\hbox{$\includegraphics[width=1.3cm]{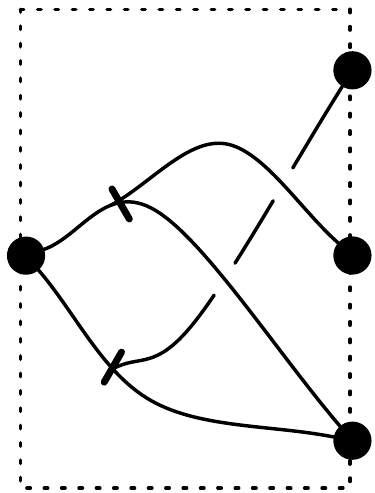}$}
\quad
\lower18pt\hbox{$\includegraphics[width=1.2cm]{ldiag}$} \ \comp\ 
\lower22pt\hbox{$\includegraphics[width=1.3cm]{diagtensorid}$}\ =\ 
\lower22pt\hbox{$\includegraphics[width=1.3cm]{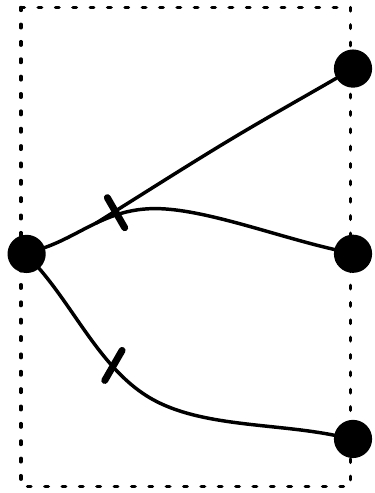}$}
\end{equation}

All linking diagrams in $\Sp{\Relfc}$ can be obtained from the basic set of components \eqref{eq:csp}, \eqref{eq:sp} and \eqref{eq:common}, combined using the operations of composition and tensor.
\begin{theorem}\label{thm:syntaxforcontention}
Every arrow in $\Sp{\Relfc}$ decomposes into an expression consisting only of $\diagsym$, $\rightEndsym$, $\codiagsym$, $\leftEndsym$, $\ldiagsym$, $\rzerosym$, $\lcodiagsym$, $\lzerosym$, $\idsym$, $\twsym$, composed with $\comp$ and $\otimes$.
\end{theorem}
\begin{proof}
Omitted.
\end{proof}


\section{Multisets and multirelations}
\label{sec:multisets}


We have seen that $\Sp{\Relfc}$ is a setting in which one can study the algebra of \eqref{eq:csp}, \eqref{eq:sp} and \eqref{eq:common}. Here we develop a second, different setting, that arises from a compositional algebra of P/T nets~\cite{Bruni2013}.

Given a set $X$, let $\multiset{X}$ denote the set of finite maps $\U: X\to \N$, ie where $dom(\U)$ is a finite set.
We call elements of $\multiset{X}$ \emph{multisets}. We will sometimes abuse set notation to when talking about multisets; any ordinary set $U\subseteq X$ can  be considered as a multiset in the obvious way: \[
U x=\begin{cases}1 & \text{if }x\in U \\ 0 & \text{otherwise.}\end{cases}
\]
Given $\U, \V\in \multiset{X}$, $\U + \V$ is the multiset $(\U + 
V)(x)=\U x + \V x$. We say $\U \geq \V$ if $\forall x.\; \U x - \V x\in \N$. If $\U \geq \V$, let $(\U - \V)\in\multiset{X}$ be defined $(\U - \V)x \Defeq \U x - \V x$. 
Given $k\in \N$ and $\U\in \multiset{X}$, $k\U(x)\Defeq k\cdot\U(x)$.

$\multiset{X}$ is the action on objects of the functor
$\multiset{-}:\Set \to \Set$. On functions,  
$\multiset{f}:\multiset{X}\to\multiset{Y}$ is defined
$\multiset{f}\U (y)=\sum_{x\in X : f(x)=y}\U x$; 
note that since $\U$ is nonzero on a finite subset of $X$, this is well-defined. 
There is a natural transformation 
$\mu_X: \multiset{\multiset{X}}\to \multiset{X}$ that
takes $\mu_X \mathcal{V} (x) = \sum_{\mathcal{V}(\mathcal{U})\geq 0} \mathcal{V}\mathcal{U}\cdot \mathcal{U}x$
and $\eta_X: X\to \multiset{X}$ where $\eta_X x (y) = \begin{cases} 1 & \text{if }x=y \\ 0 & \text{otherwise.} \end{cases}$ It is not difficult to check that $(\multiset{-},\mu,\eta)$ is a monad, commonly referred to as the multiset monad. Given $f:A\to \multiset{B}$, the definition of $f^\#:\multiset{A} \to \multiset{B}$ follows from a simple calculation: $\lift{f}(\mathcal{U}) = \sum_{a\in A} (\mathcal{U}a) f(a)$.

Let $\multirel\Defeq Kl(\multiset{-})$ and $\Relfm$ be the full subcategory of $\multirel$ with objects the finite sets. The arrows of $\Relfm$ are thus functions $f:X\to \multiset{Y}$ in $\Setf$, we will sometimes write $f:X \mrelto Y$.

\subsection{Multi synchronisations}
Suppose that $f:A\mrelto X$ and $g: B\mrelto X$ in $\Relfm$. A (multi $f$, $g$) \emph{synchronisation} is a pair $(\U,\V)$ with $\U\in\multiset{A}$ and  $\V\in\multiset{B}$ such that $\lift{f}\U = \lift{g}\V$. A synchronisation thus consists of a multiset of $A$ together with a multiset of $B$ that both map to the same multiset of $X$ via $\lift{f}$ and $\lift{g}$, respectively; this notion is the multiset equivalent of the notion of synchronisation that we have considered in \S\ref{subsec:synchronisations}. We will again write $\synch{\U}{\V}$ as shorthand and write $\msynch{f}{g}$ for the set of synchronisations.
 
Synchronisations inherit an ordering from multisets, pointwise.
If we have $\synch{\U'}{\V'} \leq \synch{\U}{\V}$ then $\synch{\U-\U'}{\V-\V'}$: indeed $\lift{f}(\U-\U') = \lift{f}\U - \lift{f}\U'
=\lift{g}\V-\lift{g}\V' = \lift{g}(\V-\V')$. Synchronisations are closed under linear combinations: if $\{\synch{\U_i}{\V_i}\}_{i\in I}$ and $k_i\in \N$ then define $\sum_i k_i\synch{\U_i}{\V_i}\Defeq (\sum_i k_i\U_i , \sum_i k_i\V_i)$, which is clearly a synchronisation.

A set $\mathbf{X}$ of synchronisations is \emph{mutually incomparable} when
\begin{multline*}
\forall \synch{\U}{\V},\synch{\U'}{\V'}\in \mathbf{X}.\; \synch{\U}{\V} \leq \synch{\U'}{\V'} \ \vee\  \synch{\U'}{\V'} \leq \synch{\U}{\V} \\ \ \Rightarrow\  \synch{\U}{\V} = \synch{\U'}{\V'}.
\end{multline*}

We need to recall a version of Dickson's lemma~\cite{Dickson1913}, stated in terms of synchronisations. It can be proved by a straightforward induction.
\begin{lem}[Dickson]\label{lem:dickson}
Suppose $f:A\mrelto X$ and $g: B\mrelto X$ in $\Relfm$.
Any set $X$ of mutually incomparable multi-$f$, $g$ synchronisations is finite. 
\end{lem}

Let $\minmsynch{f}{g}$ be the set of minimal synchronisations. Clearly any two minimal synchronisations are incomparable, thus, by the conclusion of Lemma~\ref{lem:dickson}, $\minmsynch{f}{g}$ is finite. In particular
\eqref{eq:weakpullbackdiagram} is a commutative diagram in $\Relfm$
where $p\synch{\U}{\V}=\U$ and $q\synch{\U}{\V}=\V$.
\begin{equation}\label{eq:weakpullbackdiagram}
\raise30pt\hbox{$
\xymatrix@=15pt{
& {\minmsynch{f}{g}} \ar[dl]_-p \ar[dr]^-q \\
{A} \ar[dr]_f & & {B} \ar[dl]^g \\
& {X}
}$}
\end{equation}

\subsection{Weak pullbacks in $\Relfm$}

The following result shows that any synchronisation can be written as a linear combination of minimal synchronisations. 
\begin{lem}\label{lem:decomposition}
If $\synch{\U}{\V}$ then there exists a
family $\{(k_i,\synch{\U_i}{\V_i})\}_{i\in I}$, where each $\synch{\U_i}{\V_i}$ is minimal and different from $\synch{\U_j}{\V_j}$ for all $j\neq i$, s.t.\  $\synch{\U}{\V}=\sum_{i} k_i \synch{\U_i}{\V_i}$.  The family is called a \emph{minimal decomposition} of $\synch{\U}{\V}$.
\end{lem}
\begin{proof}
Simple induction.\qed
\end{proof}

The conclusion of Lemma~\ref{lem:decomposition} implies that~\eqref{eq:weakpullbackdiagram} is a weak pullback diagram: given $\alpha: Y \mrelto A$ and $\beta:Y\mrelto B$ such that $f\alpha=g\beta$ in $\Relfm$, $h: Y\mrelto \minmsynch{f}{g}$ takes $y$ to a minimal decomposition of $\synch{\alpha y}{\beta y}$.

\begin{rmk}
The diagram~\eqref{eq:weakpullbackdiagram} is merely a weak pullback, because the decomposition of Lemma~\ref{lem:decomposition} is not, in general, unique. Indeed, consider $t:2\to 1$ in $\Relfm$ with $t0=t1=\{0\}$. Now $\minmsynch{t}{t}=\{\synch{\{0\}}{\{0\}}, \synch{\{0\}}{\{1\}}, \synch{\{1\}}{\{0\}}, \synch{\{1\}}{\{1\})}\}$. Consider $u: 1\to 2$ in $\Relfm$ with $u0=\{0,1\}$. Then $\synch{u0}{u0}$ but there are several minimal decompositions: eg
$\synch{\{0\}}{\{0\}}+\synch{\{1\}}{\{1\}}$ and
$\synch{\{0\}}{\{1\}}+\synch{\{1\}}{\{0\}}$.
\end{rmk}

\section{Linking diagrams in $\Spr{\Relfm}$}
\label{sec:multilinkingdiags}


Consider $\Spr{\Relfm}$, with objects that the natural numbers and arrows spans $k \xleftarrow{f} x \xrightarrow{g} l$ in $\Relfm$ where $x\to k\times l$ is injective\footnote{In other words, the internal binary relations in $\Relfm$: 
an internal relation is a span $k\leftarrow x \rightarrow l$ where $x \to k\t l$ is mono.}.
Composition proceeds in two steps. First, given 
\[
k_0 \xleftarrow{f_0} x_0 \xrightarrow{g_0} k_1 \xleftarrow{f_1} x_1 \xrightarrow{g_1} k_2,\]
construct
$k_0 \xleftarrow{f_0 p} \minmsynch{g_0}{f_1} \xrightarrow{g_1 q} k_2$ where $p:\minmsynch{g_0}{f_1} \to x_0$ and $q:\minmsynch{g_0}{f_1} \to x_1$ are the projections.
 In general, however, $[f_0 p, g_1 q]:\minmsynch{g_0}{f_1} \to k_0\times k_2$ may be non-injective, thus we obtain $\minmsynch{g_0}{f_1}'$ below, together with $f',g'$ in $\Relfm$ through an epi-mono factorisation of $[\lift{f_0}p,\lift{g_1}q]$ in $\Set$, and this is the composition.
\[
k_0 \xleftarrow{f'} \minmsynch{g_0}{f_1}' \xrightarrow{g'} k_2
\]
\begin{proposition}
$\Spr{\Relfm}$ is a category.
\end{proposition}
\begin{proof} (Sketch) 
The non-trivial part is showing that composition is associative. The essence is captured in the diagram below, in $\Relfm$.
\[
\xymatrix@=15pt{
& & {\minmsynch{g_0}{f_1p_1}} \ar@{.>}[rr]^{\Phi}
\ar@/_2pc/[ddl]_-{r_0} \ar[drr]^(.4){s_0} \ar@{.>}[d]^{h_0} & & 
    {\minmsynch{g_1q_0}{f_2}} \ar@{.>}[d]_{h_1}
    \ar[dll]|{\hole}_(.4){r_1} \ar@/^2pc/[ddr]^-{s_1}
 \\
& & {\minmsynch{g_0}{f_1}} \ar@{}[dd]|{(\dagger)}
\ar[dl]_-{p_0} \ar[dr]^-{q_0} & & 
{\minmsynch{g_1}{f_2}} \ar@{}[dd]|{(\ddagger)}
\ar[dl]_-{p_1} \ar[dr]^-{q_1} \\
& {x_0} \ar[dl]_{f_0}
\ar[dr]^{g_0} & & {x_1} \ar[dl]_{f_1} \ar[dr]^{g_1} & & {x_2} \ar[dl]_{f_2} \ar[dr]^{g_2} \\
{k_0} & & {k_1} & & {k_2} & & {k_3}
}
\]
In addition to the two weak pullback diagrams $(\dagger)$ and $(\ddagger)$,
we have a set $\minmsynch{g_0}{f_1 p_1}$ and the projection maps in $\Relfm$ 
\[ 
r_0:\minmsynch{g_0}{f_1 p_1} \to x_0,\ s_0: \minmsynch{g_0}{f_1 p_1} \to \minmsynch{g_1}{f_2}\]
and a set $\minmsynch{g_1 q_0}{f_2}$ together with maps
\[
r_1: \minmsynch{g_1 q_0}{f_2} \to \minmsynch{g_0}{f_1},\ s_1: \minmsynch{g_1 q_0}{f_2} \to x_2
\]
The sets $\minmsynch{g_0}{f_1 p_1}$, $\minmsynch{g_1 q_0}{f_2}$ are not, in general isomorphic, 
for similar reasons why the $\minmsynch{f}{g}$ construction fails to be a pullback; there is, in general, more than one decomposition of a synchronisation into a linear combination of minimal synchronisations.

This is not a problem, because all that we require is that 
$(f_0r_0,g_2q_1s_0)$ and $(f_0p_0r_1,g_2s_1)$ have the same
image in $\multiset{k_0}\times\multiset{k_2}$. 

To show this, first we use the weak pullback property of 
$(\dagger)$ to obtain 
$h_0:\minsynch{g_0}{f_1p_1}\to\minsynch{g_0}{f_1}$,
satisfying $p_0h_0=r_0$ and $q_0h_0=p_1s_0$.
The second of these equations, together
with the fact that $\minsynch{g_1q_0}{f_2}$
is a weak pullback allows us to obtain
\[
\Phi:\minsynch{g_0}{f_1p_1}\to
\minsynch{g_1q_0}{f_2}
\]
 that
satisfies $r_1\Phi = h_0$ and
$s_1\Phi = q_1s_0$. 
Now, for any $\sigma\in\minsynch{g_0}{f_1p_1}$ we have 
$f_0r_0\sigma = f_0p_0h_0\sigma = f_0p_0r_1\Phi\sigma$
and 
$g_2q_1s_0\sigma=g_2s_1\Phi\sigma$,
so the image of $(f_0r_0,g_2q_1s_0)$
is contained in the image of $(f_0p_0r_1,g_2s_1)$.
A symmetric argument, constructing morphisms
$h_1: \minsynch{g_1q_0}{f_2}\to
\minsynch{g_1}{f_2}$  and $\Psi: \minsynch{g_1q_0}{f_2}\to \minsynch{g_0}{f_1p_1}$ allows us to demonstrate the reverse inclusion.

\qed
\end{proof}
Note that, as indicated in the proof above, the ``relational'' requirement on spans is necessary in order to ensure associativity of composition. Again there is a tensor product inherited from the coproduct in $\Setf$.


\subsection{The algebra of $\Spr{\Relfm}$}

While we no longer have to draw contention, in $\Spr{\Relfm}$ links can have multiple connections to boundary ports. 
We indicate this by
\begin{wrapfigure}{r}{0.15\textwidth}
\includegraphics[width=1.5cm]{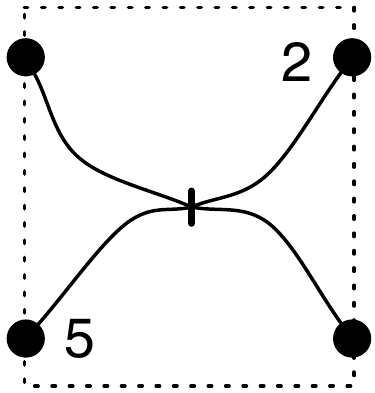}
\end{wrapfigure}
annotating connections with natural numbers $\geq 2$: for instance the diagram to the right 
is the span $2\xmrelleftarrow{a}1\xmrelrightarrow{b} 2$ where
$(a 0)(0) = (b 0)(1) = 1$, $(a 0)(1)=5$ and $(b 0)(0)=2$.

Considering the diagrams of $\eqref{eq:csp}$ and $\eqref{eq:common}$ in $\Spr{\Relfm}$, all the equations in \eqref{eq:cspcomonoid}, \eqref{eq:cspassoc}, \eqref{eq:frobenius}, \eqref{eq:separable}, \eqref{eq:compactclosed} hold in $\Spr{\Relfm}$.
On the other hand, the structure in \eqref{eq:sp} and \eqref{eq:common} satisfies the equations in \eqref{eq:spcomonoid}, \eqref{eq:spassoc}, \eqref{eq:bialgebra} and \eqref{eq:reverseunit}. Differently from \eqref{eq:diagcodiag}, in $\Spr{\Relfm}$ we have the following:
\begin{equation}\label{eq:speqs2}\tag{$\ldiagsym\lcodiagsym{}^{\mathcal{M}}$}
\lower17pt\hbox{$\includegraphics[width=1.2cm]{ldiag}$} \ \comp\ 
\lower17pt\hbox{$\includegraphics[width=1.2cm]{lcodiag}$} \ =\ 
\lower9pt\hbox{$\includegraphics[width=1.2cm]{id}$} 
\end{equation}

Below, we show how \eqref{eq:csp} and \eqref{eq:sp} interact in $\Spr{\Relfm}$.

\begin{equation}\label{eq:mdiagandzero}\tag{$\diagsym\!\!\rzerosym\!\!\lzerosym\!\!{}^\mathcal{M}$}
\lower18pt\hbox{$\includegraphics[width=1.2cm]{diag}$} \ \comp\ 
\lower18pt\hbox{$\includegraphics[width=1.2cm]{idtensorrzero}$} \ =\   \lower9pt\hbox{$\includegraphics[width=1.2cm]{rzerocomplzero}$}
\quad
\lower18pt\hbox{$\includegraphics[width=1.2cm]{codiag}$} \ \comp\  \lower9pt\hbox{$\includegraphics[width=1.2cm]{rzero}$} \ =\  \lower18pt\hbox{$\includegraphics[width=1.2cm]{rzerotensorrzero}$}
\end{equation}
\begin{equation}\label{eq:mldiagandend}\tag{$\ldiagsym\rightEndsym\leftEndsym{}^{\mathcal{M}}$}
\lower18pt\hbox{$\includegraphics[width=1.2cm]{ldiag}$} \ \comp\  \lower18pt\hbox{$\includegraphics[width=1.2cm]{idtensorrightend}$} \ =\  \lower9pt\hbox{$\includegraphics[width=1.2cm]{idandrightend}$} 
\quad
\lower18pt\hbox{$\includegraphics[width=1.2cm]{lcodiag}$} \ \comp\  \lower9pt\hbox{$\includegraphics[width=1.2cm]{rightend}$} \ =\  \lower18pt\hbox{$\includegraphics[width=1.2cm]{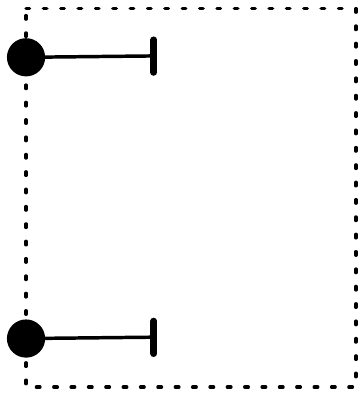}$}
\end{equation}
\begin{equation}\label{eq:mdiagandldiag}\tag{$\diagsym\lcodiagsym{}^\mathcal{M}$}
\lower18pt\hbox{$\includegraphics[width=1.2cm]{diag}$} \ \comp\ 
\lower18pt\hbox{$\includegraphics[width=1.2cm]{lcodiag}$} \ =\  \lower9pt\hbox{$\includegraphics[width=1.2cm]{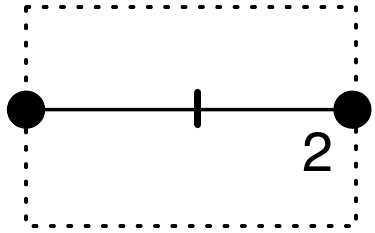}$}
\quad
\lower18pt\hbox{$\includegraphics[width=1.2cm]{lcodiag}$} \ \comp\ 
\lower18pt\hbox{$\includegraphics[width=1.2cm]{diag}$} \ =\ 
\lower18pt\hbox{$\includegraphics[width=1.2cm]{twodiags}$}
\end{equation}
\begin{equation}\label{eq:mdiagandldiag} \tag{$\diagsym\ldiagsym^\mathcal{M}$}
\lower18pt\hbox{$\includegraphics[width=1.2cm]{diag}$} \ \comp\ 
\lower22pt\hbox{$\includegraphics[width=1.3cm]{ldiagtensorid}$} \ =\  \lower22pt\hbox{$\includegraphics[width=1.3cm]{diagldiag}$}
\quad
\lower18pt\hbox{$\includegraphics[width=1.2cm]{ldiag}$} \ \comp\ 
\lower22pt\hbox{$\includegraphics[width=1.3cm]{diagtensorid}$}\ =\ 
\lower22pt\hbox{$\includegraphics[width=1.3cm]{ldiagdiag}$}
\end{equation}


The equations in~\eqref{eq:mdiagandzero} are the same as in \eqref{eq:diagandzero}. The left equation in \eqref{eq:mldiagandend}
is the same as the corresponding one in \eqref{eq:ldiagandend}, but the right hand side equations differ because the contention relation does not play a role in $\Spr{\Relfm}$. The right hand side equation in \eqref{eq:mdiagandldiag}
agrees with the corresponding one in \eqref{eq:diagandlcodiag}, but the left one deserves attention: while in \eqref{eq:diagandlcodiag} there was no possible synchronisation between $\diagsym$ and $\lcodiagsym$ because of the fact that the two links in $\lcodiagsym$ were in contention, in $\Sp{\Relfm}$ there is a synchronisation that involves all three links, as represented in the left hand side equation of \eqref{eq:mdiagandldiag}. The interaction between $\ldiagsym$ and $\ldiagsym$ is as in \eqref{eq:diagandldiag}, and \eqref{eq:mdiagandldiag} are the same as \eqref{eq:diagandldiag}. 

\begin{theorem}\label{thm:syntaxformulti}
Every arrow in $\Spr{\Relfm}$ decomposes into an expression consisting only of $\diagsym$, $\rightEndsym$, $\codiagsym$, $\leftEndsym$, $\ldiagsym$, $\rzerosym$, $\lcodiagsym$, $\lzerosym$, $\idsym$, $\twsym$, composed with $\comp$ and $\otimes$.
\end{theorem}
\begin{proof}
Omitted.
\end{proof}


\section{Conclusion}
\label{sec:conclusion}


We have studied two categories of linking diagrams. The first, $\Sp{\Relfc}$,  arose from the study of a compositional algebra of C/E nets, called C/E nets with boundaries. Indeed, the arrows of $\Sp{\Relfc}$ are just C/E nets with boundaries, without places. The second, $\Spr{\Relfm}$ arose from the study of a compositional algebra of P/T nets, called P/T nets with boundaries. The arrows of $\Spr{\Relfm}$ are P/T nets with boundaries, without places. These categories generalise previous work by Hughes~\cite{Hughes2008}.

Both categories are ``expressive enough'' to carry two different commutative monoid-comonoid structures on objects, one of which a separable Frobenius algebra, the other a commutative bialgebra. In both settings the interaction between the two structures is interesting and we have examined some of the phenomena that arise. Both categories are generated by the small number of basic components that witness the monoid-comonoid structures.

In future work a full axiomatisation will be presented, and the categories of linking diagrams will be shown to characterise the arrows of the resulting free categories. The theory of PROPs~\cite{Lack2004a} seems well adapted for expressing the relationship between the algebraic structures, as well as the complete algebras of C/E and P/T nets; Fiore and Campos~\cite{Fiore2013} have recently used a similar setting to develop the algebra of dags.

\paragraph{Acknowledgment.}
Thanks to R.F.C.\;Walters for inspiration and guidance, and to the referees for  remarks that have improved the presentation.
{
\tiny
\bibliography{jab}
\bibliographystyle{abbrv}
}
\end{document}